\documentclass[envcountsame,envcountsect]{llncs}
\usepackage{amsmath, amssymb}
\usepackage[dvipsnames]{xcolor}
\usepackage[a4paper,top=3cm,bottom=2.5cm,left=3cm,right=3cm,marginparwidth=1.75cm]{geometry}
\usepackage{xspace}
\usepackage{dsfont}
\usepackage{dsfont}
\usepackage[colorlinks=true,urlcolor=blue,linkcolor=ForestGreen,filecolor=blue,citecolor=LimeGreen,pdfpagemode=UseOutlines,pdfstartview=FitH,pdfpagelayout=OneColumn,bookmarks=true]{hyperref}
\usepackage{fullpage}
\usepackage{enumerate}
\usepackage{graphicx}
\usepackage{algorithm, algpseudocode, float}
\usepackage[T1]{fontenc}
\usepackage[utf8]{inputenc}
\DeclareUnicodeCharacter{200B}{~}
\usepackage{thmtools}
\usepackage{thm-restate}

\DeclareMathAlphabet{\eurm}{U}{eur}{m}{n}
\DeclareMathAlphabet{\mathbsf}{OT1}{cmss}{bx}{n}\DeclareMathAlphabet{\mathssf}{OT1}{cmss}{m}{sl}\DeclareMathAlphabet{\mathcsf}{OT1}{cmss}{sbc}{n}

\DeclareSymbolFont{bsfletters}{OT1}{cmss}{bx}{n}  
\DeclareSymbolFont{ssfletters}{OT1}{cmss}{m}{n}
\DeclareMathSymbol{\bsfGamma}{0}{bsfletters}{'000}
\DeclareMathSymbol{\ssfGamma}{0}{ssfletters}{'000}
\DeclareMathSymbol{\bsfDelta}{0}{bsfletters}{'001}
\DeclareMathSymbol{\ssfDelta}{0}{ssfletters}{'001}
\DeclareMathSymbol{\bsfTheta}{0}{bsfletters}{'002}
\DeclareMathSymbol{\ssfTheta}{0}{ssfletters}{'002}
\DeclareMathSymbol{\bsfLambda}{0}{bsfletters}{'003}
\DeclareMathSymbol{\ssfLambda}{0}{ssfletters}{'003}
\DeclareMathSymbol{\bsfXi}{0}{bsfletters}{'004}
\DeclareMathSymbol{\ssfXi}{0}{ssfletters}{'004}
\DeclareMathSymbol{\bsfPi}{0}{bsfletters}{'005}
\DeclareMathSymbol{\ssfPi}{0}{ssfletters}{'005}
\DeclareMathSymbol{\bsfSigma}{0}{bsfletters}{'006}
\DeclareMathSymbol{\ssfSigma}{0}{ssfletters}{'006}
\DeclareMathSymbol{\bsfUpsilon}{0}{bsfletters}{'007}
\DeclareMathSymbol{\ssfUpsilon}{0}{ssfletters}{'007}
\DeclareMathSymbol{\bsfPhi}{0}{bsfletters}{'010}
\DeclareMathSymbol{\ssfPhi}{0}{ssfletters}{'010}
\DeclareMathSymbol{\bsfPsi}{0}{bsfletters}{'011}
\DeclareMathSymbol{\ssfPsi}{0}{ssfletters}{'011}
\DeclareMathSymbol{\bsfOmega}{0}{bsfletters}{'012}
\DeclareMathSymbol{\ssfOmega}{0}{ssfletters}{'012}

\newcommand{\calA}{{\mathcal{A}}}
\newcommand{\calB}{{\mathcal{B}}}

\newcommand{\calD}{{\mathcal{D}}}
\newcommand{\calE}{{\mathcal{E}}}

\newcommand{\calG}{{\mathcal{G}}}

\newcommand{\calI}{{\mathcal{I}}}

\newcommand{\calK}{{\mathcal{K}}}
\newcommand{\calL}{{\mathcal{L}}}
\newcommand{\calM}{{\mathcal{M}}}
\newcommand{\calN}{{\mathcal{N}}}

\newcommand{\calT}{{\mathcal{T}}}
\newcommand{\calS}{{\mathcal{S}}}

\newcommand{\calX}{{\mathcal{X}}}
\newcommand{\calY}{{\mathcal{Y}}}

\newcommand{\E}[2][]{{\mathbb{E}_{#1}}{\left(#2\right)}}       
\renewcommand{\P}[2][]{{\textnormal{Pr}_{#1}}{\left[#2\right]}}

\newcommand{\card}[1]{\ensuremath{\left|{#1}\right|}}           \newcommand{\abs}[1]{\ensuremath{\left|#1\right|}}              \newcommand{\norm}[2][]{\ensuremath{{\left\Vert{#2}\right\Vert}_{#1}}}   \newcommand{\eqdef}{\ensuremath{\triangleq}}                      \newcommand{\indic}[1]{\ensuremath{\mathds{1}\!\left\{#1\right\}}}

\renewcommand{\leq}{\leqslant}
\renewcommand{\geq}{\geqslant}

\newcommand{\tr}[1]{\ensuremath{\text{\textnormal{tr}}\left(#1\right)}}

\newcommand{\argmax}{\mathop{\text{argmax}}}

\newcommand{\pr}[1]{\left({#1}\right)}
\newcommand{\ket}[1]{|#1\rangle}
\newcommand{\bra}[1]{\langle #1 |}
\newcommand{\braket}[2]{\langle #1 | #2 \rangle}
\newcommand{\id}{\textnormal{id}}

\newcommand{\one}{\mathbf{1}}

\newcommand{\N}{\mathbb{N}}

\newcommand{\set}[1]{{\left\{#1\right\}}}
\newcommand{\ketbra}[2]{{\ket{#1}\bra{#2}}}
\newcommand{\kb}[1]{{\ket{#1}\bra{#1}}}

\newcommand{\Dens}[1]{\mathcal{D}\!\pr{#1}}
\newcommand{\Unit}[1]{\textnormal{U}\!\pr{#1}}
\newcommand{\Lin}[1]{\calL\!\pr{#1}}

\newcommand{\pg}[5][]{\omega^{#1}_{\textnormal{#2}}(#3|#4)_{#5}}

\newcommand{\pwinunif}[2]{\textnormal{p}_{\textnormal{win-unif}} \pr{#1;#2}}
\newcommand{\pwinunifs}[1]{\textnormal{p}_{\textnormal{win-unif}}^* \pr{#1}}
\newcommand{\pwinind}[2]{\textnormal{p}_{\textnormal{win-ind}} \pr{#1;#2}}
\newcommand{\pwininds}[1]{\textnormal{p}_{\textnormal{win-ind}}^* \pr{#1}}

\newcommand{\rank}[1]{{\textnormal{rank}}\pr{#1}}

\newcommand{\KeyGen}{\mathsf{KeyGen}}
\newcommand{\Enc}{\mathsf{Enc}}
\newcommand{\Dec}{\mathsf{Dec}}
\mathchardef\mhyphen="2D

\newcommand{\haarscheme}[2]{\mathsf{HMB\mhyphen QECM}\pr{#1;#2}}
\newcommand{\haarschemeu}[1]{\mathsf{HMB\mhyphen QECM}\pr{#1}}
\newcommand{\Haar}{\mathsf{Haar}}

\newcommand{\needspace}[2]{#2}

\newcommand{\proj}[1]{\ket{#1}\!\bra{#1}}

\newcommand{\negl}{\mathrm{negl}}

\title{Limitations on Uncloneable Encryption and\\
  Simultaneous One-Way-to-Hiding}
\author{Christian Majenz \inst{1,2,4}, Christian Schaffner \inst{2,3}, Mehrdad Tahmasbi \inst{2,3} 
\institute{
	Centrum Wiskunde \& Informatica (CWI), Amsterdam, Netherlands \and 
	QuSoft, Amsterdam, Netherlands \and
	Institute for  Logic, Language and Computation (ILLC), University of Amsterdam, Amsterdam, Netherlands \and
	Department of Applied Mathematics and Computer Science, Technical University of Denmark, Kgs. Lyngby, Denmark
	\\ \email{christian.majenz@cwi.nl}, \email{c.schaffner@uva.nl}, \email{m.tahmasbi@uva.nl}
}}

\begin{document}
\maketitle
\begin{abstract}
We study uncloneable quantum encryption schemes for classical messages as recently proposed by Broadbent and Lord~\cite{broadbent2019uncloneable}. We focus on the information-theoretic setting and give several limitations on the structure and security of these schemes: Concretely, 1) We give an explicit cloning-indistinguishable attack that succeeds with probability $\frac12 + \mu/16$ where $\mu$ is related to the largest eigenvalue of the resulting quantum ciphertexts. 2) For a uniform message distribution, we partially characterize the scheme with the minimal success probability for cloning attacks. 3) Under natural symmetry conditions, we prove that the rank of the ciphertext density operators has to grow at least logarithmically in the number of messages to ensure uncloneable security. 4) The \emph{simultaneous} one-way-to-hiding (O2H) lemma is an important technique in recent works on uncloneable encryption and quantum copy protection. We give an explicit example which shatters the hope of reducing the multiplicative ``security loss'' constant in this lemma to below 9/8.
\end{abstract}

\section{Introduction}

The linearity of quantum mechanics, leading to features like the ``quantum no-cloning theorem'' and the imperfect distinguishability of non-orthogonal quantum states, has opened up several opportunities in cryptography. In fact, the very first protocols in quantum information processing where cryptographic applications capitalizing on the mentioned features of quantum theory, like Wiesner's quantum money  \cite{Wiesner83} and the BB84 key exchange protocol \cite{BB84} (for a survey of other cryptographic applications, see~\cite{BS16}). Another application of the fact that quantum information cannot be copied, called \emph{uncloneable encryption}, was recently explored by Broadbent and Lord~\cite{broadbent2019uncloneable}: Using a secret key, Alice encrypts a classical message into a \emph{quantum} ciphertext, which is passed to a cloner, Eve, who copies the ciphertext into two quantum registers. These quantum registers are then provided to two separated parties, Bob and Charlie. Alice also provides the secret key to Bob and Charlie who attempt to guess the message. The adversaries Eve, Bob, and Charlie win if and only if Bob and Charlie both correctly decrypt the message. Given a message drawn from some (possibly adversarially chosen) distribution $p$ over a set of size $M$, the adversaries can always win with the maximum probability that $p$ assigns to any message, by outputting the same fixed message as their guess.  The goal is to devise an encryption scheme such that 1) a legitimate receiver can recover the message from the ciphertext and the secret key; 2) the probability that the adversaries win does not significantly exceed the trivial success probability. Following \cite{broadbent2019uncloneable}, we call the latter requirement \emph{uncloneable-security}, which is a strengthening of the confidentiality notion using the same message distribution (e.g. one-wayness if $p$ is the uniform distribution, and IND if  $p$ is uniform on an adversarially chosen pair of messages).\footnote{Gottesman studied a different definition of uncloneable encryption called \emph{quantum tamper-detection}~\cite{G03}.} Studying this notion of security is motivated by cryptographic applications such as quantum money and the prevention of storage attack by classical adversaries\footnote{In this scenario, an eavesdropper Eve without quantum memory is forced to immediately measure any intercepted quantum ciphertexts, without knowing the key (yet). Uncloneable security implies that these classical measurement outcomes cannot be used to determine the message even after learning the key, because otherwise, copying these classical outcomes and handing them to Bob and Charlie would violate uncloneable security.}~\cite{broadbent2019uncloneable}. Furthermore, uncloneable encryption is fundamentally related to quantum copy-protection \cite{Aaronson09,ALLZZ20,ALP20,CMP20}. 

The adversaries can win with probability one if the ciphertext is classical and zero-error decryption of the ciphertext is possible, by simply copying the classical ciphertext and share copies between Bob and Charlie. Broadbent and Lord~\cite{broadbent2019uncloneable} constructed two schemes with quantum ciphertexts. First, they studied the scheme in which each bit is randomly encoded in a BB84 basis determined by a secret key. The optimal probability of winning was shown to be $\pr{\frac{1}{2} + \frac{1}{2\sqrt{2}}}^n$ when the message is $n$ random bits. Second, they constructed a scheme based on a random oracle to which Alice, Bob, and Charlie have quantum access. They prove that when a message is uniformly distributed over a set of size $M$ and the number of queries made by Bob and Charlie is polynomially bounded, the optimal probability of winning is bounded by $\frac{9}{M} + \negl(\lambda)$ where $\lambda$ is a security parameter.

The security proof of the scheme introduced in~\cite{broadbent2019uncloneable}, as well as the security proof of the copy-protection scheme presented in \cite{CMP20}, are based on a \emph{``simultaneous''} variant of the so-called one-way-to-hiding (O2H) lemma, originally introduced by Unruh \cite{Unruh15} (see \cite{AHU19,BHHHP19,KSSSS20} for variations and improvements).  A variant of  that result implies that given a quantum algorithm $\calA$ having quantum oracle access to a random function $H:\calX\to \calY$, the probability that the algorithm correctly finds $H(x)$ for a fixed input $x\in \calX$ is upper-bounded by $	\frac{1}{\card{\calY}} + q\sqrt{p}$ 
where $q$ is the number of queries made by $\calA$ and $p$ is the probability of obtaining $x$ when measuring the input register of the oracle for a randomly chosen query. In the ``simultaneous'' version of this problem, two non-communicating parties with shared entanglement,  run quantum algorithms $\calA$ and $\calB$ with quantum oracle access to the same random function $H:\calY \to \calX$. In \cite[Lemma 21]{broadbent2019uncloneable}, it has been shown that the probability that both algorithms correctly output $H(x)$ for a fixed $x$ is upper-bounded by
\needspace{$9/|\calY| + \text{poly}(q_{\calA}, q_{\calB}) \sqrt{p}$,}{\begin{align}
	\frac{9}{\card{\calY}} + \text{poly}(q_{\calA}, q_{\calB}) \sqrt{p},
\end{align}}
where $q_{\calA}$ and $q_{\calB}$ are the number of queries made by $\calA$ and $\calB$, respectively, poly is a polynomial, and $p$ is the probability that measuring the input registers of both algorithms at two independently chosen queries returns $x$ on both sides.

\subsection{Our Contributions}
In this article, we explore the fundamental information-theoretic limits of uncloneable encryption. In particular, we prove the following four results.
\begin{enumerate}
    \item When the message is chosen uniformly at random to be either an adversarially chosen message or a default one, we construct in Section~\ref{sec:simultaneousguessing} an explicit cloning attack with probability of success $\frac{1}{2} + \frac{\mu}{16}$ where $\mu$ is the maximum (over all messages) of the average (with respect to the key) of the maximum eigenvalue of the corresponding ciphertext (see Corollary~\ref{cor:ind-converse}). 
This bound implies that in order to have the probability of success for all attacks limited to at most $\frac{1}{2} + \negl(\lambda)$ (as desired in \cite[Definition 11]{broadbent2019uncloneable}), the quantity $\mu$ should be negligible in $\lambda$.
\item When the message is uniformly distributed over all messages, we show in Section~\ref{sec:uncloneablesecurity} that there exists an attack with the probability of winning $\Omega\pr{\needspace{(\log M /(dM))^{1/2}}{{\frac{\log M }{d}}}}$ where $M$  is the number of messages and $d$ is the dimension of the Hilbert space corresponding to ciphertexts (see Theorem~\ref{th:lower-bound-unif}). This lower bound implies that to ensure that the probability of winning is $O\pr{\frac{1}{M}}$ for all adversaries (as desired in \cite[Definition 8]{broadbent2019uncloneable}), the rank of the ciphertext density operators has to grow at least as $\Omega(\log M)$.
\item  Fixing the number of messages $M$ and the dimension of the ciphertext Hilbert space, we partially characterize  the encryption scheme that minimizes the optimal adversarial winning probability when the message is uniformly distributed over all possible messages (Section~\ref{sec:optimalscheme}). Our characterization involves an optimization over all probability distributions over a finite set.\needspace{}{\\}
  In addition, we formulate a natural conjecture that the optimal probability distribution is in fact deterministic.\item  Finally, in Section~\ref{sec:counterexample}, we employ some of the insights from 1. \needspace{}{above} in the context of the simultaneous O2H lemma. Here, an important open question is whether the factor $9$ is an artifact of the proof technique used in \cite{broadbent2019uncloneable}\needspace{}{, or whether a probability of success of
$\frac{1}{\card{\calY}} + \text{poly}(q_{\calA}, q_{\calB}) \sqrt{p}$ is, in fact, possible}? Put differently, do there actually exist algorithms $\calA$ and $\calB$ that simultaneously succeed in guessing $H(x)$ with non-trivial probability while not allowing simultaneous extraction of $x$ from their queries? 

In this work, we answer the above question. We provide an example with $p=0$ (so simultaneous query-based extraction never succeeds), $\calY = \set{0, 1}$ but $\calA$ and $\calB$ both output $H(0)$ with probability $9/16$, which is strictly larger than the trivial $\frac12$. This example illustrates that the simultaneous setting is fundamentally different from the single-party setting.

\end{enumerate}

\section{Notation and Preliminaries}
Let $\N$ denote the set of positive integers. For $n\in \N$, $[n]$ denotes the set $\set{0, \ldots, n-1}$.

With some abuse of notation, we denote the Hilbert space corresponding to the quantum system $A$ by $A$ itself. $\card{A}$ denotes the dimension of $A$. We also denote the tensor product of $A$ and $B$ by $AB$. $\one_A$ is the identity operator over $A$. Let $\Lin{A}$ denote the set of all linear operators from $A$ to $A$, $\Dens{A}$ denote the set of all density operators over $A$ and $\Unit{A}$ denote the set of all unitary operators over $A$. A quantum channel from the quantum system $A$ to the quantum system $B$ is a linear trace-preserving completely positive map from $\Lin{A}$ to $\Lin{B}$. A positive operator-valued measure (POVM) over quantum system $A$ is a collection $\set{P_x}_{x\in \calX}$, where $\calX$ is a finite set, $P_x$ is a positive operator in $\Lin{A}$ for all $x\in \calX$, and $\sum_{x\in \calX} P_x = \one_A$. For a density matrix $\rho_{A} \in \Dens{A}$, $\lambda_{\max}(\rho_A)$ denotes the maximum eigenvalue of $\rho_A$.

The \emph{uniformly spherical measure} and \emph{Haar measure} are defined over the unit sphere in $A$ and $\Unit{A}$, respectively, as in \cite[Chapter 7]{Watrous2018}.

\subsection{Uncloneable Encryption} \label{sec:definitions}
We recall
the definition of an encryption scheme that encrypts a classical message and a classical key to a quantum ciphertext. 
\begin{definition}[\cite{broadbent2019uncloneable}, Definition~4]
  A \emph{quantum encryption of classical messages} scheme (QECM) is a triplet of algorithms $\mathcal{E} = (\KeyGen, \Enc, \Dec)$ described as follows:
  \begin{itemize}
  \item The key generation algorithm $\KeyGen$ samples a classical key $k \in \mathcal{K}$ from the key space $\calK$ with distribution $P_K$.
  \item The encryption algorithm $\Enc_{k}(m)$ takes as inputs the classical key $k$ and classical message $m \in \mathcal{M}$ from the message set $\mathcal{M}$ and produces a quantum ciphertext $\rho_A \in \Dens{A}$.
    \item The decryption algorithm $\Dec_{k}(\rho_A)$ takes as inputs the classical key $k$ and quantum ciphertext $\rho_A$ and returns the classical message $m \in \mathcal{M}$.
    \end{itemize}
\end{definition}
We note that our definition of a QECM differs slightly from~\cite{broadbent2019uncloneable}. In particular, we do not include a security parameter in our definition of a QECM, because we only study information-theoretic security in this article and therefore do not impose any computational assumptions on the adversary. Our results hold for any fixed underlying parameter of the scheme.

Correctness is defined in the natural way.
\begin{definition}
\label{def:correct}
  A QECM $\mathcal{E} = (\KeyGen, \Enc, \Dec)$ is \emph{(perfectly) correct} if for all $k$ produced by $\KeyGen$, and for all messages $m \in \mathcal{M}$, it holds that
  \begin{align}
    \Pr[\Dec_k(\Enc_k(m)) = m] = 1 \, .
  \end{align}
\end{definition}

As in~\cite{broadbent2019uncloneable}, we study two flavors of uncloneable security.  

\subsubsection{Uncloneable Security for Uniformly Distributed Messages}
A uniformly distributed message is encrypted using a secret key. The adversary then ``clones'' the ciphertext into two quantum registers and passes each register on to a separate party, called Bob and Charlie. Bob and Charlie then learn the secret key and attempt to individually decrypt the message. They are successful if they simultaneously decrypt the correct value of the message.

\begin{definition}[Cloning Attack, \cite{broadbent2019uncloneable}, Definition~7]
Let $\mathcal{E} = (\KeyGen, \Enc, \Dec)$ be a QECM scheme. A \emph{cloning attack} against $\mathcal{E}$ is a triple $\mathcal{A}=(\mathcal{N}_{A \to BC},\{P_m^k\},\{Q_m^k\})$ such that 
    \begin{itemize}
    \item The quantum channel $\calN_{A \to BC}$ describes the adversary's cloning operation.
      \item For every possible key $k$, $\set{P_m^k}_{m\in \calM}$ is Bob's POVM on $B$ to guess the message $m$.
      \item For every possible key $k$, $\set{Q_m^k}_{m\in \calM}$ is Charlie's POVM on $C$ to guess the message $m$.
\end{itemize}
The success probability of a cloning attack $\calA$ against encryption scheme $\calE$ with uniform messages is
\begin{align}
    \pwinunif{\mathcal{E}}{\mathcal{A}} \eqdef \frac{1}{\card{\calM}}\sum_m \E[k\leftarrow \KeyGen]{\tr{P_m^k\otimes Q_m^k \mathcal{N}_{A\to BC}(\Enc_k(m))}}\
\end{align}
\end{definition}    

We also define the optimal probability of winning as
\begin{align}
    \pwinunifs{\calE} \eqdef \needspace{ {\sup}_{\mathcal{A}}\,}{ \sup_{\mathcal{A}}} \pwinunif{\mathcal{E}}{ \mathcal{A}}. \,
\end{align}
where the supremum is taken over all cloning attacks $\calA$. 

\subsubsection{Uncloneable-Indistinguishable Security}
Here, the adversary chooses a message $m_1\in \calM$. The encrypted message is uniformly distributed over the set $\set{m_0, m_1}$ for a fixed $m_0\in \calM$. The rest of the definition of an attack $\calA$ is similar to in the previous section. 

\begin{definition}[Cloning-indistinguishability  Attack, \cite{broadbent2019uncloneable}, Definition~10]
Let $\mathcal{E} = (\KeyGen, \Enc, \Dec)$ be a QECM scheme and fix $m_0 \in \calM$. A \emph{cloning-indistinguishability attack} against $\mathcal{E}$ and $m_0$ is a quadruple $\mathcal{A}=(m_1, \mathcal{N}_{A \to BC},\{P_b^k\},\{Q_b^k\})$ such that 
    \begin{itemize}
    \item $m_1$ is a message in $\calM\setminus \set{m_0}$.
    \item The quantum channel $\calN_{A \to BC}$ describes the adversary's cloning operation.
      \item For every possible key $k$, $\set{P_b^k}_{b\in\set{0, 1}}$ is Bob's POVM on $B$ to guess the message $m_b$.
      \item For every possible key $k$, $\set{Q_b^k}_{b\in\set{0, 1}}$ is Charlie's POVM on $C$ to guess the message $m_b$.
\end{itemize}
The \emph{success probability of a cloning-indistinguishability attack $\calA$ against encryption scheme $\calE$} is
\begin{align}
    \pwinind{\calE}{\calA} \eqdef \frac{1}{2}\needspace{{\sum}_{b \in \{0,1\}}\,}{\sum_{b \in \{0,1\}}}\E[k\leftarrow \KeyGen] {\tr{P_b^k\otimes Q_b^k \mathcal{N}_{A\to BC}(\Enc_k(m_b))}}.
\end{align}
\end{definition}    

We also define the optimal probability of winning as
\begin{align}
    \pwininds{\mathcal{E}} \eqdef \needspace{ {\sup}_{\mathcal{A}}\,}{ \sup_{\mathcal{A}}}\pwinind{\mathcal{E}}{ \mathcal{A}}. \, 
\end{align}

\section{Simultaneous Guessing and Uncloneable-Indistinguishable Security} \label{sec:simultaneousguessing}
We consider the situation where two parties Bob and Charlie have quantum side information about a classical random variable $X$ belonging to set $\calX$, and they simultanously try to guess $X$ by local measurements.

\begin{definition}
For a classical-quantum-quantum (cqq) state $\rho_{XBC} =\sum_{x\in \calX} P_X(x) \kb{x} \otimes \rho_{BC}^x$, the \emph{simultaneous guessing probability} is defined as 
\needspace{\begin{align}
    \pg{X}{B}{C}{\rho} \eqdef {\sup}_{\set{P_x}_{x\in \calX}, \set{Q_x}_{x\in \calX}} {\sum}_{x \in \calX} P_X(x) \: \tr{(P_x \otimes Q_x) \rho_{BC}^x}
\end{align}}
{\begin{align}
		\pg{}{X}{B;C}{\rho} \eqdef \sup_{\set{P_x}_{x\in \calX}, \set{Q_x}_{x\in \calX}} \sum_{x \in \calX} P_X(x) \: \tr{(P_x \otimes Q_x) \rho_{BC}^x}
\end{align}}
\end{definition}

One of our central objects of study in this article is the following ``cloning operation:''
\begin{align} \label{eq:cloning}
    V_{A \to BC}: \ket{\phi} \mapsto\needspace{ \pr{\ket{\bot}_B \otimes \ket{\phi}_C + \ket{\phi}_B \otimes \ket{\bot}_C}/\sqrt 2}{ \frac{1}{\sqrt{2}} \pr{\ket{\bot}_B \otimes \ket{\phi}_C + \ket{\phi}_B \otimes \ket{\bot}_C}},
\end{align}
where $\ket{\bot}$ is a unit vector orthogonal to $A$. Intuitively, $V_{A\to BC}$ \needspace{sends $\ket\psi$ to $B$ or $C$}{distributes the input state to $B$ and $C$} ``in superposition.''

Let $\rho, \sigma$ be perfectly distinguishable states, for example $\rho=\proj{0}, \sigma=\proj{1}$. We now consider the task where for a random bit $X \in \{0,1\}$, Bob and Charlie have to simultaneously  distinguish the \needspace{cases where $V \rho V^\dagger$  ($X=0$), or $V \sigma V^\dagger$ ($X=1$) is handed to Bob and Charlie.}{following two cases:
\begin{itemize}
\item if $X=0$: $V \rho V^\dagger$ is handed to Bob and Charlie, 
\item if $X=1$: $V \sigma V^\dagger$ is handed to Bob and Charlie.
\end{itemize}}
The following lemma gives a non-trivial lower bound on their simultaneous guessing probability of $X$\needspace{}{ for this task}. In particular, for pure states like $\rho=\proj{0}, \sigma=\proj{1}$, we obtain a lower bound of $\frac12 + \frac{1}{16}=\frac 9{16}$. At first sight, it seems counterintuitive that Bob and Charlie are able to succeed with probability strictly higher than $\frac12$\needspace{: A}{. One might think that a}fter applying the cloning operation $V_{A \to BC}$, the state \needspace{should be}{is} either with Bob or with Charlie, so the other party will succeed with probability at most $\frac12$. However, as one can see from the explicit simultaneous guessing strategy that we construct in the proof of the lemma, Bob and Charlie can exploit the quantum coherence of the state after applying $V_{A \to BC}$ to achieve a simultaneous guessing probability strictly larger than $\frac12$.

\begin{lemma}
\label{lm:pg-lower-bound}
Let $\rho, \sigma \in \Dens{A}$ such that $\rho\sigma= 0$ and define $\tau_{XBC} \eqdef\frac{1}{2} \kb{0}_X \otimes V\rho V^\dagger +\frac{1}{2} \kb{1}_X \otimes V\sigma V^\dagger$. We have
\begin{align}
    \pg{}{X}{B;C}{\tau} \geq \frac{1}{2} + \frac{\max(\lambda_{\max}(\rho), \lambda_{\max}(\sigma))}{16}
\end{align}
\end{lemma}

\begin{proof}
Let $|A| = d$. We consider eigen-decompositions 
\begin{align}
    \rho = \sum_{i\in[d]} \lambda_i \kb{a_i}, \qquad \sigma = \sum_{i\in[d]} \mu_i \kb{b_i}, 
\end{align}
such that $\lambda_0 \geq \cdots \geq \lambda_{d-1}$ and $\mu_0 \geq \cdots \geq \mu_{d-1}$. We set $\ket{\phi} \eqdef \sqrt{1-\alpha} \ket{a_0} + \sqrt{\alpha}\ket{\bot}$ for some parameter $\alpha \in [0, 1]$ (to be determined below) and
\begin{align}
    \Pi \eqdef \kb{\phi} +  \needspace{\sum_{i\in[|A|]\setminus\{1\}:\lambda_i >0}}{\sum_{i\in[d]\setminus\set{0}}} \kb{a_i}.
\end{align}
$\Pi$ is a projector and one can verify the following equalities by straightforward calculations:
\needspace{\begin{align}
		\bra{\bot}\Pi \ket{\bot} &= \alpha \label{eq:pi-a-calc1},\quad
		\bra{a_0}\Pi \ket{a_0} = 1-\alpha,\quad 
		\bra{a_0} \Pi \ket{\bot} =\bra{\bot} \Pi \ket{a_0} = \sqrt{\alpha(1-\alpha)}, \text{ as well as}\\
		\bra{a_i}\Pi \ket{a_i} &= 1 \text{ and }
		\bra{a_i} \Pi \ket{\bot} =\bra{\bot} \Pi \ket{a_i} = 0 \quad \forall i\in [d]\setminus \set{0} \text{ such that } \lambda_i > 0\label{eq:pi-a-calc5}
\end{align}}
{\begin{align}
    \bra{\bot}\Pi \ket{\bot} &= \alpha \label{eq:pi-a-calc1}\\
    \bra{a_0}\Pi \ket{a_0} &= 1-\alpha\\
    \bra{a_0} \Pi \ket{\bot} &=\bra{\bot} \Pi \ket{a_0} = \sqrt{\alpha(1-\alpha)}\\
    \bra{a_i}\Pi \ket{a_i} &= 1 \quad \forall i\in [d]\setminus \set{0} \text{ such that } \lambda_i > 0\\
    \bra{a_i} \Pi \ket{\bot} &=\bra{\bot} \Pi \ket{a_i} = 0 \quad \forall i\in [d]\setminus \set{0} \text{ such that } \lambda_i > 0.\label{eq:pi-a-calc5}
\end{align}}
It holds that $\lambda_i \mu_j \abs{\braket{a_i}{b_j}}^2 = 0$ for all $i$ and $j$ since $\rho\sigma= 0$. Hence $\Pi \ket{b_j} = 0$ for all $j$ with $\mu_j > 0$ and 
\needspace{\begin{align}
    \bra{b_j}(\one - \Pi)\ket{b_j} &= 1,\quad\quad \bra{b_j}(\one-\Pi)\ket{\bot} = \bra{\bot}(\one-\Pi)\ket{b_j} = 0\label{eq:pi-b-calc2},
\end{align}}
{\begin{align}
		\bra{b_j}(\one - \Pi)\ket{b_j} &= 1\label{eq:pi-b-calc1}\\ 
		\bra{b_j}(\one-\Pi)\ket{\bot} &= \bra{\bot}(\one-\Pi)\ket{b_j} = 0\label{eq:pi-b-calc2},
\end{align}}
 for all $j$ with $\mu_j > 0$.

Bob and Charlie both use the POVM $\{\Pi, \one - \Pi\}$ as their local guessing strategies for $X$. By definition of $\pg{}{X}{B;C}{\tau}$, we have
\begin{align}
    \pg{}{X}{B;C}{\tau} \geq \frac{1}{2}\pr{\tr{(\Pi \otimes \Pi)  V   \rho V^\dagger} + \tr{(\one - \Pi) \otimes (\one - \Pi) V\sigma V^\dagger} }.\label{eq:pg-pi}
\end{align}
The first term on the right-hand side of Eq.\eqref{eq:pg-pi} is
\needspace{\begin{align}
    &\tr{(\Pi \otimes \Pi)  V\rho V^\dagger}
    = \sum_{i\in[|A|]}\lambda_i\tr{(\Pi \otimes \Pi) V\kb{a_i}V^\dagger}= \frac{1}{2}\sum_{i\in[|A|]}\lambda_i\textnormal{tr} \Big( (\Pi \otimes \Pi) \big( \kb{a_i} \otimes \kb{\bot}\\
    & \!\!\!\!\!\!\! + \!  \kb{\bot} \!\otimes\! \kb{a_i} \! + 
  \ketbra{a_i}{\bot} \!\otimes\!\ketbra{\bot}{a_i} \! + \!  \ketbra{\bot}{a_i} \!\otimes\!\ketbra{a_i}{\bot}   \big) \Big)
   \!\stackrel{(a)}{=}\! 2\lambda_1\alpha(1\!-\!\alpha) \! +\! \alpha \!\!\!\!\!\!\needspace{\sum_{i\in[|A|]\setminus\{1\}}}{\sum_{i=2}^{|A|}} \!\!\!\!\!\!\lambda_i\!=\! \alpha \! +\! \lambda_1 \alpha (1\!-\!2\alpha),\label{eq:pi-tr}
\end{align}}
{\begin{align}
		\tr{(\Pi \otimes \Pi)  V\rho V^\dagger}
		&= \sum_{i\in[d]}\lambda_i\tr{(\Pi \otimes \Pi) V\kb{a_i}V^\dagger}\\
		\begin{split} &= \frac{1}{2}\sum_{i\in[d]}\lambda_i\textnormal{tr} \Big( (\Pi \otimes \Pi) \big( \kb{a_i} \otimes \kb{\bot} \! + \,  \kb{\bot} \otimes \kb{a_i} \, + \\
			& \qquad \qquad    \ketbra{a_i}{\bot} \otimes\ketbra{\bot}{a_i} \, + \,  \ketbra{\bot}{a_i} \otimes\ketbra{a_i}{\bot}   \big) \Big)
		\end{split} \\
		&\stackrel{(a)}{=} 2\lambda_1\alpha(1-\alpha)  + \alpha \needspace{\sum_{i\in[d]\setminus\{1\}}}{\sum_{i\in [d] \setminus \set{0}}} \lambda_i\\
		&= \alpha + \lambda_0 \alpha (1-2\alpha),\label{eq:pi-tr}
\end{align}}
where $(a)$ follows by using Eq.~\eqref{eq:pi-a-calc1}-\eqref{eq:pi-a-calc5}.
Similarly applying Eq.~\needspace{\eqref{eq:pi-b-calc2}}{\eqref{eq:pi-b-calc1} and \eqref{eq:pi-b-calc2}} yields that 
\begin{align}
    \tr{((\one-\Pi) \otimes (\one-\Pi))  V\sigma  V^\dagger} = 1 - \alpha. \label{eq:one-minus-pi-tr}
\end{align}
Combining Eq.~\eqref{eq:pg-pi}, \eqref{eq:pi-tr}, and \eqref{eq:one-minus-pi-tr} and setting $\alpha \eqdef 1/4$, we obtain that
\begin{align}
    \pg{}{X}{B;C}{\tau} \geq \frac{1}{2} + \frac{\lambda_0}{16}.
\end{align}
Finally, \needspace{WLOG}{without loss of generality} we can assume that $\lambda_0 \geq \mu_0$, and therefore, $\lambda_0 = \max(\lambda_{\max}(\rho), \lambda_{\max}(\sigma))$.
\qed\end{proof}

Applying the above lemma to \needspace{}{the setting of }uncloneable-indistinguishable encryption, we obtain the following\needspace{}{ corollary}. 
\begin{corollary}
\label{cor:ind-converse}
For any correct (see Definition~\ref{def:correct}) QECM scheme $\calE$ and arbitrary $m_0 \in \calM$, there exists an uncloneable-indistinguishable attack $\calA$ against $\calE$ and $m_0$ for which it holds that
\begin{align}
     \pwinind{\mathcal{E}}{ \mathcal{A}} \geq \frac{1}{2} + \frac{\max_{m \in \calM} \E[k\leftarrow \KeyGen]{ \lambda_{\max}(\Enc_k(m)) }}{16}
\end{align}
\end{corollary}
\begin{proof}
  Let $m_1 \eqdef \argmax_{m \in \calM \setminus \{m_0\}}  \E[k\leftarrow \KeyGen]{ \lambda_{\max}(\Enc_k(m)) }$.
We consider the unclonable-indistinguishable attack $(m_1, V_{A \to BC}, \{ \Pi^k, \one-\Pi^k \}, \{ \Pi^k, \one-\Pi^k \})$ where the projector $\Pi^k$ is $\Pi$ defined in the proof of Lemma~\ref{lm:pg-lower-bound} for $\rho= \Enc_k(m_0)$ and $\sigma = \Enc_k(m_1)$. The claim then follows directly from the lemma.
\qed \end{proof}

\section{An Optimal Scheme} \label{sec:optimalscheme}
In this section, we provide a partial answer to the following question: For QECM schemes with finite message set $\calM = [M]$ and a $d$-dimensional ciphertext space $A$, which QECM scheme $\calE = (\KeyGen, \Enc, \Dec)$ minimizes $\pwinunifs{\calE}$?
It turns out that the best QECM schemes in terms of uncloneable security are of a simple form, formally defined in Definition~\ref{def:optimal-scheme}. Intuitively, the optimal scheme maps every classical message to a completely mixed state over a subspace of $A$, so that different messages are mapped to states with disjoint support to ensure correctness. The key additionally specifies a Haar-random unitary to hide the message.

For simplicity of the proof, we work with Haar-random unitaries in this section, which is a continuous distribution and results in infinite-sized keys. In practice, one would want to work with finite key sizes and pick the unitary from a suitably chosen two-design instead.

\begin{definition}
\label{def:optimal-scheme}
Fix an orthonormal basis $a = (\ket{a_0}, \cdots, \ket{a_{d-1}})$ for $A$ and a random variable $T$ taking values in $\calT  \eqdef \set{(t_0, \cdots, t_{M-1}) \in \N^M: \sum_{m\in[M]} t_m = d}$. We define the \emph{Haar measure-based QECM} scheme $\haarscheme{T}{M, d}= (\widetilde{\KeyGen}, \widetilde{\Enc}, \widetilde{\Dec})$ as follows. The key generation $\widetilde{\KeyGen}$ outputs a pair $K=(T, U)$ where  $U$ is a random unitary over $A$ distributed according to the Haar measure and independent of $T$. Furthermore, we define
\needspace{\begin{align}
		\widetilde{\Enc}_{((t_0, \cdots, t_{M-1}),u)}(m) \eqdef u \pr{\frac{1}{t_m} {\sum}_{i=\sum_{j=0}^{m-1}t_j-1}^{\sum_{j=0}^{m}t_j} \kb{a_i}} u^\dagger \, .
	\end{align}}
{\begin{align}
    \widetilde{\Enc}_{((t_0, \cdots, t_{M-1}),u)}(m) \eqdef u \pr{\frac{1}{t_m} \sum_{i=\sum_{j=0}^{m-1}t_j - 1}^{\sum_{j=0}^{m}t_j - 1} \kb{a_i}} u^\dagger \, .
\end{align}}
Decryption $\widetilde{\Dec}$ is defined by applying $u^\dagger$, measuring in basis $a$ and identifying the message $m$. Note that the choice of the orthonormal basis $a$ does not affect the performance of the protocol by the invariance of Haar measure, and therefore, we drop $a$ from our notation.\\
Furthermore, when $d = LM$ for an integer $L$ and $T = (L, \cdots, L)$ with probability one, we denote $\haarscheme{T}{M, d}$ by $\haarschemeu{M, d}$. 
\end{definition}

The following theorem shows that QECM schemes of the form above are optimal in terms of uncloneable security for uniform messages. \begin{theorem}
\label{th:optimal-scheme}
For any correct (see Definition~\ref{def:correct}) QECM scheme $\calE = (\KeyGen, \Enc, \Dec)$, it holds that
\begin{align}\label{eq:opt-scheme}
    \pwinunifs{\calE} \geq \inf_T \pwinunifs{\haarscheme{T}{M,d}} \geq \pwinunifs{\haarschemeu{M, M(d-M+1)}},
\end{align}
where the infimum is taken over all random variable $T$ that are permutation invariant, i.e. $\P{T = (t_1, \cdots, t_M)} = \P{T = (t_{\sigma(1)}, \cdots, t_{\sigma(M)})}$ for all $(t_1, \cdots, t_M) \in \calT$ and all permutations $\sigma$.
We also have for all $1\leq M' \geq M-1$,
\begin{align}
    \pwinunifs{\calE} \geq \frac{M'}{M}\pwinunifs{\haarschemeu{M', \max\pr{d, \frac{M'd}{M-M'}}}}
\end{align}
\end{theorem}
To prove the above theorem, we shall introduce four modifications of an arbitrary correct QECM scheme and discuss how they affect uncloneable security.
\begin{enumerate}
    \item \textbf{``Uniformization'' w.r.t. Haar measure:} We associate to any correct scheme $\calE$ a scheme $\haarscheme{T}{M, d}$ as specified in Definition~\ref{def:optimal-scheme}, for which we only need to define random variable $T$. For a fixed key value $k$, we choose $t = (t_0, \cdots, t_{M-1})\in \calT$ such that $t_m\geq \rank{\Enc_k(m)}$ for all. \emph{Note that when $\sum_{m\in [M]} \rank{\Enc_k(m)} < d $, there is no unique such choice of $(t_0, \cdots, t_{M-1})$. Then, we choose $t$ arbitrarily among all valid choices.} Generating the key $k$ according to $\KeyGen$ defines a random variable $T$.  We further define another random variable  $\widetilde{T}$ as 
\begin{align}
    \P{\widetilde{T} = (t_0, \cdots, t_{M-1})} = \frac{1}{M!}\sum_{\sigma: \text{permutation of } [M]} \P{T = (t_{\sigma(0)}, \cdots, t_{\sigma(M-1)})}.
\end{align}
We then have the following result about the uncloneable security of $\calE$, $\haarscheme{T}{M, d}$, and $\haarscheme{\widetilde{T}}{M, d}$.
\begin{lemma}
\label{lm:unif-perf}
Random variable $\widetilde{T}$ is permutation invariant and
\begin{align}
    \pwinunifs{\calE} \geq \pwinunifs{\haarscheme{T}{M, d}}=\pwinunifs{\haarscheme{\widetilde{T}}{M, d}}
\end{align}
\end{lemma}
    We provide the full proof in Appendix~\ref{sec:unif-perf-proof}. The sketch of the proof has the following two steps. First, for the QECM $\calE$, we augment the key with a random unitary $U$ and apply it to the ciphertext in the encryption procedure. This still results in a correct QECM with improved uncloneable security. When $U$ is distributed according to Haar measure, the basis that ciphertexts were initially encoded in will be forgotten, and only their spectrum matters. Second, we apply a random permutation of the eigenvectors of each ciphertext. This extra randomness is not necessary to decrypt the message but makes the spectrum of each ciphertext flat.
\item \textbf{Extension of ciphertext space:} Let $A'$ be any Hilbert space with $|A'| \geq |A|$ and $V_{A\to A'}$ be an arbitrary isometric. We define a new scheme $\calE' = (\KeyGen', \Enc', \Dec')$ with the same message size $M$ 
\begin{align}
    \KeyGen' &= \KeyGen\\
    \Enc'_k(m) &= V \Enc_k(m) V^\dagger\\
    \Dec'_k(\rho) &= \Dec_k(V^\dagger\rho V).
\end{align}
$\calE'$ is correct and we have $\pwinunifs{\calE} = \pwinunifs{\calE'}$, but as the ciphertext space  has larger dimension, we have more flexibility to choose $T$ in the first modification. In particular, if we have $\rank{\Enc_k(m)} \leq r$ for all $m$ and $k$, by choosing $A'$ of dimension $rM$, we can set $t$ to be (independently of $k$) $(r, \cdots, r)$ for $\calE'$.
\item \textbf{Expurgation of ciphertexts:} Let $M' \leq M$ and $\phi_k: [M']\to [M]$ be an injective function for all $k$. Define $\calE' = (\KeyGen', \Enc', \Dec')$ with message size $M'$ as 
\begin{align}
\KeyGen' &= \KeyGen\\
\Enc'_k(m) &= \Enc_k(\phi_k(m))\\
\Dec'_k(\rho) &= \phi_k^{-1}(\Dec_k(\rho)).
\end{align}
Note that 1) $\calE'$ is correct. 2) $ \frac{M'}{M} \pwinunifs{\calE'} \leq \pwinunifs{\calE} \leq \pwinunifs{\calE'} $ 3) For any $M' \leq M$, there exists a choice of functions $\phi_k$ such that $\rank{\Enc'_k(m)} \leq \frac{d}{M - M'}$ for all $k$ and all $m \in [M']$.
\end{enumerate}

\begin{proof}
We refer here to above modifications as \emph{first, second, and third modification}, respectively. According to first modification and Lemma~\ref{lm:unif-perf}, there exists permutation invariant random variable $T$ such that 
\begin{align}
    \pwinunifs{\calE} \geq \pwinunifs{\haarscheme{T}{M, d}}.
\end{align}
Next we perform the second modification for Hilbert space $A'$ with $|A'| = M(d - M + 1)$ to obtain a new scheme $\calE'$. For this new scheme, we can choose $T = (d-M+1, \cdots, d-M+1)$ as $\rank{\Enc'_k(m)} = \rank{\Enc_k(m)} \leq d - M + 1$ for all $k$ and $m$ and all correct QECMs. Therefore,
\begin{align}
    \pwinunifs{\calE} = \pwinunifs{\calE'} \geq \pwinunifs{\haarschemeu{M, M(d-M+1)}},
\end{align}
which completes the proof of Eq.~\eqref{eq:opt-scheme}.

We next apply the third modification with arbitrary $M'$ to obtain a scheme $\calE' = (\KeyGen', \Enc', \Dec')$ and choose $\phi_k$ such that $\rank{\Enc'_k(m)} \leq \frac{d}{M - M'}$ for all $k$ and $m\in[M']$. We next apply the second modification to $\calE'$ to obtain another scheme $\calE'' = (\KeyGen'', \Enc'', \Dec'')$ for a Hilbert space $A'$ with $|A'| = \max\pr{d, \frac{M'd}{M-M'}}$. In the first modification of $\calE''$, we can choose $T = (\frac{d}{M - M'}, \cdots, \frac{d}{M - M'})$ because $\rank{\Enc''_k(m)} = \rank{\Enc'_k(m)} \leq \frac{d}{M - M'}$. Therefore,
\begin{align}
    \pwinunifs{\calE} \geq \frac{M'}{M}\pwinunifs{\calE'} = \frac{M'}{M}\pwinunifs{\calE''} \geq \frac{M'}{M}\pwinunifs{\haarschemeu{M', \max\pr{d, \frac{M'd}{M-M'}}}},
\end{align}
as claimed in Theorem~\ref{th:optimal-scheme}.
\qed \end{proof}
We remark that if  the decryption algorithm, which can be specified by a POVM on the ciphertext space, is projective, but the QECM scheme is only \emph{approximately correct}, Equation \eqref{eq:opt-scheme} can be proven as well, up to an additive error proportional to the square root of the correctness error. This is the case because such a scheme can be made perfectly correct by modifying the encryption algorithm to  prepare the post-measurement state of a successful decryption measurement instead of the original ciphertext, which only differs from the original encryption algorithm up to a small error bounded by the gentle measurement lemma  \cite{Winter99}.

Theorem~\ref{th:optimal-scheme} results in equivalent conditions for obtaining optimal asymptotic unconeable security, summarized in the following two corollaries.

\begin{corollary}
For any fixed $M$, the following statements are equivalent.
\begin{enumerate}
    \item There exists a sequence $\set{\calE_\lambda}_{\lambda \geq 1}$ of QECMs with message size $M$ such that $\lim_{\lambda\to \infty}\pwinunifs{\calE_\lambda} = \frac{1}{M}$.
    \item $\lim_{d\to \infty} \pwinunifs{\haarschemeu{M, d}} = \frac{1}{M}$.
\end{enumerate}
\end{corollary}

\begin{corollary}
The following statements are equivalent.
\begin{enumerate}
    \item There exists a sequence $\set{\calE_\lambda}_{\lambda \geq 1}$ of QECMs  such that $\lim_{\lambda\to \infty}\pwininds{\calE_\lambda} = \frac{1}{2}$.
    \item $\lim_{d\to \infty} \pwinunifs{\haarschemeu{2, d}} = \frac{1}{2}$.
\end{enumerate}
\end{corollary}

We leave it as an open problem to characterize the optimal probability distribution of $T$, but we conjecture that a deterministic $T$ that splits the space evenly is optimal.
\begin{conjecture}
\label{conj:optimal}
We have
\begin{align}
\inf_T \pwinunifs{\haarscheme{T}{M,d}} = \pwinunifs{\haarschemeu{M, d}}.
\end{align}
\end{conjecture}

\section{Uncloneable Security for Uniformly Distributed Message} \label{sec:uncloneablesecurity}
We prove a lower bound on $\pwinunifs{\calE}$ for QECM schemes $\calE$ whose ciphertexts have small rank.
\begin{theorem}
\label{th:lower-bound-unif}
Let $\calE$  be a  correct (see Definition~\ref{def:correct}) QECM scheme such that the message size is $M$ and the ciphertexts belong to a $d$-dimensional Hilbert space. 
Then there exists an absolute constant $c>0.02285$ such that
    \begin{align}
        \pwinunifs{\calE} \geq c \frac{\log M - 1}{d}.
    \end{align}
\end{theorem}

\begin{proof}
Note first that by Theorem~\ref{th:optimal-scheme}, we have
\begin{align}
    \pwinunifs{\calE} \geq \frac{1}{2} \pwinunifs{\haarschemeu{\frac12 M, d}}.
\end{align}
Therefore, in the rest of proof we can assume that 
\begin{align}
     \Enc_k(m) = \frac{\Pi_k^m}{d/M} \label{eq:flat-spec-assum}
\end{align}
and only replace $M$ by $\frac12M$ and multiply the lower-bound on $\pwinunifs{\calE}$ by $\frac12$ at the end.

We consider a specific attack for the adversaries described as follows. Let $\ket{e_1}, \cdots, \ket{e_d}$ be an orthonormal basis for $A$ and define $P_i \eqdef \kb{e_i}$ for all $i$. The cloner performs the POVM $\set{P_i}$ on the ciphertext and share the classical output $i$ with both Bob and Charlie who decode the message as $\widehat{m}(i, k)$ using $i$ and the key $k$. The probability of winning for this attack is 
\begin{align}
    \max_{\widehat{m}}\P[m, k, i]{\widehat{m}(i, k) = m}
    &=  \max_{\widehat{m}}\E[k\sim P_K]{\frac{1}{M}\sum_{m} \sum_{i} \tr{P_i \Enc_k(m)}\indic{\widehat{m}(i, k)= m} }\\
    &= \max_{\widehat{m}}\E[k\sim P_K]{\frac{1}{M} \sum_{i} \tr{P_i \Enc_k(\widehat{m}(i, k))} }\\
    &= \E[k\sim P_K]{\frac{1}{M} \sum_{i} \max_m \tr{P_i \Enc_k(m)} },
\end{align}
where every maximization of $\widehat{m}$ is over all functions from $[d]\times [K] \to [M]$ and the maximization of $m$ is over all messages in $[M]$. For any unitary operator $U$ acting on $A$, $\set{UP_iU^\dagger}$ is a POVM. Choosing $U$ at random according to Haar measure we obtain that
\begin{align}
    \pwinunifs{\calE} 
    &\geq \E[U\sim \Haar]{\E[k\sim P_K]{\frac{1}{M} \sum_{i} \max_{m\in [M]} \tr{UP_i U^\dagger \Enc_k(m)} }} \\
    &= \frac{1}{M} \sum_{i\in [d]} {\E[k\sim P_K]{\E[U\sim \Haar]{\max_{m\in [M]} \tr{UP_i U^\dagger \Enc_k(m)} }}} \\
    &\stackrel{(a)}{=} \frac{d}{M} {\E[k\sim P_K]{\E[U\sim \Haar]{\max_{m\in [M]} \tr{UP_0 U^\dagger \Enc_k(m)} }}},
\end{align}
where $(a)$ follows from unitary invariance of Haar measure. We define next $\ket{\phi} = U\ket{e_0}$ which is a random vector distributed according to uniform spherical measure. Then, 
\begin{align}
    \E[U\sim \Haar]{\max_{m\in [M]} \tr{UP_0 U^\dagger \Enc_k(m)} } 
    &= \E[\ket{\phi}]{\max_{m\in [M]} \bra{\phi} \Enc_k(m)\ket{\phi}} \\
    &\stackrel{(a)}{=} \frac{M}{d} \E[\ket{\phi}]{\max_{m\in [M]} \bra{\phi}\Pi_k(m)\ket{\phi}}\\
    &= \frac{M}{d} \E[\ket{\phi}]{\max_{m\in [M]} \norm{\Pi_k(m)\ket{\phi}}^2}
\end{align}
where $(a)$ follows from our additional assumption in Eq.~\eqref{eq:flat-spec-assum} that $\Enc_k(m) = \Pi_k(m)/(d/M)$.
Hence, it is enough to show that $\E[\ket{\psi}]{  \max_{m\in[M]}\norm{\Pi_{k}(m) \ket{\psi}}^2 } \geq 0.0457 \frac{\log M}{d}$ for any fixed $k\in[K]$. Since $\calE$ is correct, there exists an orthonormal basis $(\ket{e_0}, \cdots, \ket{e_{d-1}})$ such that $\Pi_{k}(m) = \sum_{i\in \calI_m}  \kb{e_i}$ where $(\calI_1, \cdots, \calI_M)$ forms a partition of $[d]$.  Let $(a_0, \cdots, a_{d-1})$ and  $(b_0, \cdots, b_{d-1})$ be two independent sequences of iid standard normal random variables. Define 
\begin{align}
    \ket{\widetilde{\psi}} &= \sum_{i\in[d]} (a_i + jb_i)\ket{e_i}\\
    \ket{{\psi}} &= \ket{\widetilde{\psi}} / \norm{\ket{\widetilde{\psi}} }
\end{align}
Then, $\ket{\psi}$ is distributed according to uniformly spherical measure. We also have
\begin{align}
    \norm{\Pi_{k}(m)\ket{\psi}}^2 = \frac{ \sum_{i\in \calI_m} (a_i^2 + b_i^2)}{\sum_{i=1}^d (a_i^2 + b_i^2)}
\end{align}
Let $(X_0, \cdots, X_{M-1})$ be independent random variables such that $X_m\sim\textnormal{Erlang}(\tr{\Pi_k(m)}, 1/2)$ (See Appendix~\ref{sec:erlang-dist} for the definition and properties of Erlang distribution).  We then have
\begin{align} 
    \E[\ket{\psi}]{  \max_{m\in[M]}\norm{\Pi_{k}(m) \ket{\psi}}^2 } 
    &= \E{\frac{\max_{m\in [M]} X_m}{\sum_{m'\in [M]}X_{m'}}}\\
    &\stackrel{(a)}{\geq} 0.0457\frac{\log M}{\sum_{m\in[M]} \tr{\Pi_m(k)}}\\
    &= 0.0457\frac{\log M}{d},
\end{align}
where $(a)$ follows from Lemma~\ref{lm:erlang-max}.
\end{proof}

\section{Counter-Example for Simultaneous O2H} \label{sec:counterexample}
The so-called one-way-to-hiding (O2H) lemma \cite{Unruh15} is an important tool in the analysis of the quantum random-oracle model. Informally, it states that if an algorithm has an advantage over random guessing in determining which of two quantum-accessible oracles it has query access to, there is a reduction that outputs an input on which the two oracles differ. The latter process is called \emph{extraction}.
In \cite{broadbent2019uncloneable}, a \emph{simultaneous} variant of this lemma has been presented (Lemma 21). In this setting, the starting point is two non-communicating agents that receive (in general entangled) quantum inputs and interact with an oracle. If the two agents simultaneously succeed in producing an output of the oracle corresponding to an input where the two possible oracles differ, then that input can be simultaneously extracted by each agent.  In \cite{CMP20}, another variant of the simultaneous O2H lemma was shown. Here, like in the single-party O2H lemma, it is only required that the two agents distinguish the two possible oracles. The simultaneous success probability in the distinguishing task, however, needs to be close to 1 to guarantee a non-trivial success probability forextraction. It is an interesting open question whether a simultaneous analogue of Unruh's O2H lemma exists that gives a non-trivial extraction guarantee whenever two agents as described above simultaneously succeed at distinguishing two oracles.

In the following, we provide a counterexample, answering the above question in the negative. More precisely, for a random function $H: \{0,1\}\to\{0,1\}$, we exhibit an input state $\ket{\psi}_{BC}$ and algorithms $\mathcal A_B$ and $\mathcal A_C$ such that the following holds. $i)$ When provided with the registers $B$, and $C$, of $\ket\psi$, respectively, as input, the two algorithms both output $H(0)$ simultaneously with probability $>1/2$, and $ii)$ the two O2H extractors never succeed simultaneously.

Before presenting the counterexample, let us formally state the simultaneous O2H lemma that is proven in \cite{broadbent2019uncloneable}.
\begin{lemma}\label{lem:simO2H}
	For $L \in \{B, C\}$, let $q_L$ be a nonnegative integer, $U_L$ a unitary and
	$\{\pi^y_L\}_{y \in \{0,1\}^n}$ be a projective measurement. Let further $\ket{\psi} $ be a unit vector and $x \in \{0,1\}^\lambda$.
Then, we have
\begin{equation}
	\mathbb E_H
	\norm{
		\Pi^{H(x)}
		\left(
		\left(U_B O^H_B\right)^{q_B}
		\otimes
		\left(U_C O^H_C\right)^{q_C}
		\right)
		\ket{\psi}
	}^2
	\leq
	\frac{9}{2^n}
	+
	(3q_B q_C + 2)q_Bq_C
	\sqrt{M}\needspace{\text{ \quad and}}{}\label{eq:simO2H}
\end{equation}
\needspace{}{where $\Pi^{H(x)} = \pi_B^{H(x)} \otimes\pi_C^{H(x)}$ and}
\begin{equation}
	M
	=
	\mathbb E_k \mathbb E_\ell \mathbb E_H
	\norm{
		\left(\ketbra{x}{x}_{B_Q} \otimes \ketbra{x}{x}_{C_Q}\right)
		\left(
		\left(U_B O_B^H\right)^k
		\otimes
		\left(U_C O_C^H\right)^\ell
		\right)
		\ket{\psi}
	}^2
\end{equation}
with \needspace{$\Pi^{H(x)} = \pi_B^{H(x)} \otimes\pi_C^{H(x)}$, }{} and $k$, $\ell$, and $H$ being uniformly distributed over $\{0, \ldots, q_B-1\}$, $ \{0, \ldots, q_C-1\}$, and $\set{h:\{0,1\}^\lambda\to\{0,1\}^n}$, respectively.
\end{lemma}
The important thing to note here is, that the trivially achievable left-hand side for Equation  \eqref{eq:simO2H},  that does not require querying the oracle, is $2^{-n}$. The inequality in Equation \eqref{eq:simO2H}, however, provides a non-trivial extraction guarantee only once the left-hand side is strictly larger then $9\cdot 2^{-n}$. This discrepancy also prevents a straight-forward generalization to a full search-to-decision-style simultaneous O2H lemma, and causes the weak bound in the non-trivial generalization in this direction presented in \cite{CMP20}. Note also that the above lemma can be generalized to apply to algorithms that guess only a function of the output $H(x)$, see Lemma 19 in \cite{CMP20}. Therefore, a version of the above theorem with the factor 9 replaced by a constant $c<2$ would directly imply a search-to-decision variant.

\medskip
We continue to present our counterexample that shows that Lemma \ref{lem:simO2H} has no chance of being true when replacing the factor 9 with any constant $c<\frac{9}{8}=1.125$.
\begin{theorem}
	For $n=1$ and $q_B = q_C = 1$, there exist $U_L$ and $\pi_L$ for $L \in \{B,C\}$ such that, with the notation from Lemma \ref{lem:simO2H}, \needspace{$M=0$ but}{}
	\begin{equation}
		\mathbb E_H
		\norm{
			\Pi^{H(x)}
			\left(
			U_B O^H_B
			\otimes
			U_C O^H_C
			\right)
			\ket{\psi}
		}^2
		=\frac{9}{16}\needspace{.}{,}
	\end{equation}
\needspace{}{but $M=0$.}
\end{theorem}
\begin{proof}
Let the registers $B$ and $C$ have two qubits each, where the first one corresponds to the input to $H$, the output of $H$ will be XORred to the second qubit. Define $\ket\psi_{BC}=\frac{1}{\sqrt{2}}(\ket{0}\ket{0}_B\otimes\ket{1}\ket{+}_C+\ket{1}\ket{+}_B\otimes \ket{0}\ket{0}_C)$, where $\ket +=\frac{1}{\sqrt 2}(\ket 0+\ket 1)$, and set $U_B=U_C=\one$. After the two queries, the joint state is
\begin{equation}
	\ket{\psi_1}_{BC}=\left(
	O^H_B
	\otimes
 O^H_C
	\right)
	\ket{\psi}=\frac{1}{\sqrt{2}} \big(\ket{0}\ket{H(0)}_B\otimes\ket{1}\ket{+}_C+\ket{1}\ket{+}_B\otimes \ket{0}\ket{H(0)}_C \big).
\end{equation}
Setting $\ket\phi=\ket{0}\ket{H(0)}$ and $\ket\bot=\ket{1}\ket{+}$, we observe that
\needspace{$
	\ket{\psi_1}_{BC}=V_{A\to BC}\ket\phi_A \,
      $}
  {\begin{equation}
  		\ket{\psi_1}_{BC}=V_{A\to BC}\ket\phi_A \,
  \end{equation}}
      where $V_{A \to BC}$ was defined in Equation~\eqref{eq:cloning}.
We employ Lemma \ref{lm:pg-lower-bound} with $\rho=\proj{00}$ and $\sigma=\proj{01}$ to conclude that there exist $\pi_L^b$ for $L \in \{B,C\}$, $b=0,1$ such that
	\begin{equation}
	\mathbb E_H
	\norm{
		\Pi^{H(0)}
		\left(
		U_B O^H_B
		\otimes
		U_C O^H_C
		\right)
		\ket{\psi}
	}^2
	=\frac 1 2+\frac{1}{16}=\frac{9}{16}.
\end{equation}
However, the extraction measurement never succeeds. Indeed, there is only one query to ``choose from'' on each side, and the computational basis measurement of the first qubits of  $B$ and $C$ of the query input state $\ket\psi_{BC}$ defined above never returns $(0,0)$, hence $M=0$.
\qed\end{proof}

\section{Open Problems}
One major question is whether there exists a sequence of QECMs $\set{\calE_\lambda}_{\lambda\in \N}$ such that 
\begin{align}
    \lim_{\lambda \to \infty} \pwininds{\calE_\lambda} = \frac{1}{2} \text{ or }     \lim_{\lambda \to \infty} {\card{\calM_\lambda}\pwinunifs{\calE_\lambda}}{} = 1.
\end{align}
Analyzing the performance of the scheme defined in Definition~\ref{def:optimal-scheme} might answer this question. Another open question is the validity of Conjecture~\ref{conj:optimal}, which simplifies the characterization of the optimal scheme and justifies the assumptions that we made in Theorem~\ref{th:lower-bound-unif}. 

Regarding simultaneous O2H, while our results show that the constant $9$ on the right-hand side of Eq.~\eqref{eq:simO2H} cannot be replaced by $1$, finding the optimal constant could be of interest both for uncloneable encryption and quantum copy-protection applications. Extensions of simultaneous O2H to more than two parties are natural to study and might also have natural applications in quantum copy protection.

\section*{Acknowledgements}
We would like to thank Michael Walter for useful discussions. CM was funded by a NWO VENI grant (Project No. VI.Veni.192.159). MT and CS were supported by a NWO VIDI grant (Project No. 639.022.519).

 \needspace{\bibliographystyle{abbrv}}{\bibliographystyle{alpha}}
\newcommand{\etalchar}[1]{$^{#1}$}

\vfill\pagebreak

\appendix 
\section{Proof of Lemma~\ref{lm:unif-perf}}
\label{sec:unif-perf-proof}
Due to correctness of the QECM $\calE$, the ciphertext density matrices $\Enc_k(0), \ldots, \Enc_k(M-1)$  are mutually orthogonal for a fixed $k$. Thus, there exists an orthonormal basis $(\ket{e_0}, \cdots, \ket{e_{d-1}})$, real numbers $\lambda_{0}, \cdots, \lambda_{d-1} \in [0, 1]$, and disjoint partition $\calS_0, \cdots, \calS_{M-1}$ of $[d]$, all depending on $k$ such that
\begin{align}
    \Enc_k(m) = \needspace{{\sum}_{i\calS_m}}{\sum_{i\in \calS_m}} \lambda_{i} \kb{e_i}.
\end{align} 
 Without loss of generality, we can assume that $\calS_m = \set{i:  \sum_{j=1}^{m-1} \card{\calS_j} \leq i \leq \sum_{j=1}^m \card{\calS_j} - 1}$. According to our definition of $t$, we have $t = (\card{\calS_1}, \cdots, \card{\calS_M})$. 

 For notational simplicity, we first introduce a shorthand for simultaneous guessing as follows. For a function $g: [M] \to \Dens{BC}$ (think of it as encryption followed by a cloning operation), we define
\begin{align}
    \Delta\pr{g} \eqdef \sup_{\set{P_m}_{m\in [M]}, \set{Q_m}_{m\in [M]}} \frac{1}{M} \sum_{m\in [M]} \tr{(P_m \otimes Q_m)g(m)},
\end{align}
where the supremum is taken over all pairs of POVMs $\set{P_m}_{m\in \calM}, \set{Q_m}_{m\in \calM}$.
 We can write
\begin{align}
    \pwinunifs{\calE} = \sup_{B, C, \calN_{A\to BC}} \E[k\leftarrow \KeyGen]{ \Delta\pr{\calN_{A\to BC} \circ \Enc_k }}. \label{eq:win-delta}
\end{align}
For a fixed unitary $u$, we denote by $\calM_u$ the quantum channel mapping $\rho$ to $u\rho u^\dagger$. We have
\begin{align}
    \sup_{B, C, \calN_{A\to BC}} \E[k\leftarrow \KeyGen]{ \Delta\pr{\calN_{A\to BC} \circ \Enc_k }} = \sup_{B, C, \calN_{A\to BC}} \E[k\leftarrow \KeyGen]{ \Delta\pr{\calN_{A\to BC} \circ \calM_u \circ \Enc_k }}
\end{align}
because $\calN_{A\to BC} \mapsto \calN_{A\to BC} \circ \calM_u$ is bijection. Therefore, if $U$ is distributed according to Haar measure over $\Unit{A}$, it holds that
\begin{align}
    \sup_{B, C, \calN_{A\to BC}} \!\!\!\!\E[k\leftarrow \KeyGen]{ \Delta\pr{\calN_{A\to BC} \circ \Enc_k }}  
    &= \E[U]{\sup_{B, C, \calN_{A\to BC}}\!\! \E[k\leftarrow \KeyGen]{ \Delta\pr{\calN_{A\to BC} \circ \calM_U \circ \Enc_k }}}\\
    &\geq  \sup_{B, C, \calN_{A\to BC}} \!\!\E[k\leftarrow \KeyGen]{\E[U]{ \Delta\pr{\calN_{A\to BC} \circ \calM_U \circ \Enc_k }}}. \label{eq:haar-delta}
\end{align}
 Let $\Lambda_m$ be the set of permutations $\pi$ of $[d]$ such that $\pi(i) = i$ for all $i \notin \calS_m$. We denote by $V_\pi$ the  unitary on $A$  corresponding to the permutation $\pi$ defined by 
\needspace{$
    V_\pi \ket{e_i} = \ket{e_{\pi(i)}}.
$}
{\begin{align}
		V_\pi \ket{e_i} = \ket{e_{\pi(i)}}.
\end{align}}
Let $\pi_m$ be uniformly distributed over $\Lambda_m$ for all $m\in \calM$.
For any fixed unitary $v$ and for $U$ distributed according to Haar measure over $\Unit{A}$, $Uv$ is also distributed according to Haar measure. We hence have
\needspace{\begin{align}
    &\E[U]{ \Delta\pr{\calN_{A\to BC} \circ \calM_U \circ \Enc_k }} 
    = \E[\pi_1, \cdots, \pi_{M}]{\E[U]{  \Delta\pr{\calN_{A\to BC} \circ \calM_{UV_{\pi_1} \cdots V_{\pi_M}} \circ \Enc_k }}}\\
    &= \E[\pi_1, \cdots, \pi_{M}]{\E[U]{  \Delta\pr{\calN_{A\to BC} \circ \calM_{U}\circ\calM_{V_{\pi_1}} \circ \cdots \circ \calM_{V_{\pi_M}} \circ \Enc_k }}}\\
    &\stackrel{(a)}{\geq} {\E[U]{  \Delta\pr{\calN_{A\to BC} \circ \calM_{U}\circ\pr{\E[\pi_1, \cdots, \pi_{M}]{\calM_{V_{\pi_1}} \circ \cdots \circ \calM_{V_{\pi_M}} \circ \Enc_k}} }}}. \label{eq:delta-perm-bas}
\end{align}}
{\begin{align}
		\E[U]{ \Delta\pr{\calN_{A\to BC} \circ \calM_U \circ \Enc_k }} 
		&= \E[\pi_1, \cdots, \pi_{M}]{\E[U]{  \Delta\pr{\calN_{A\to BC} \circ \calM_{UV_{\pi_1} \cdots V_{\pi_M}} \circ \Enc_k }}}\\
		&= \E[\pi_1, \cdots, \pi_{M}]{\E[U]{  \Delta\pr{\calN_{A\to BC} \circ \calM_{U}\circ\calM_{V_{\pi_1}} \circ \cdots \circ \calM_{V_{\pi_M}} \circ \Enc_k }}}\\
		&\stackrel{(a)}{\geq} {\E[U]{  \Delta\pr{\calN_{A\to BC} \circ \calM_{U}\circ\pr{\E[\pi_1, \cdots, \pi_{M}]{\calM_{V_{\pi_1}} \circ \cdots \circ \calM_{V_{\pi_M}} \circ \Enc_k}} }}}. \label{eq:delta-perm-bas}
\end{align}}
where $(a)$ follows from the convexity of $\Delta(\cdot)$.
For a fixed $m\in \calM$, we have
\needspace{\begin{align}
    &\underset{\pi_1, \cdots, \pi_{M}}{\mathbb E}\left((\calM_{V_{\pi_1}} \!\circ\! \cdots \!\circ\! \calM_{V_{\pi_M}} \!\circ\! \Enc_k)(m)\right)
    \stackrel{(b)}{=} \underset{\pi_m}{\mathbb E}\left(\calM_{\pi_m}(\Enc_k(m))\right)= \underset{\pi_m}{\mathbb E}\left( \sum_i \lambda_{m, i}V_{\pi_m} \kb{e_i} V_{\pi_m}^\dagger\right)\\
    &=\E[\pi_m]{ \sum_i \lambda_{m, i} \kb{e_{\pi_m(i)}}} = \E[\pi_m]{ \sum_i \lambda_{m, \pi_m^{-1}(i)} \kb{e_{i}}}=  \sum_i \E[\pi_m]{ \lambda_{m, \pi_m^{-1}(i)}} \kb{e_{i}}\\
    &= \sum_{i\in \calS_m} \frac{1}{\card{\calS_m}} \kb{e_{i}}, \label{eq:expected-dens}
\end{align}}
{\begin{align}
		\E[\pi_1, \cdots, \pi_{M}]{(\calM_{V_{\pi_1}} \circ \cdots \circ \calM_{V_{\pi_M}} \circ \Enc_k)(m)} 
		&\stackrel{(b)}{=} \E[\pi_m]{\calM_{V_{\pi_m}}(\Enc_k(m))}\\
		&= \E[\pi_m]{ \sum_i \lambda_{m, i}V_{\pi_m} \kb{e_i} V_{\pi_m}^\dagger}\\
		&=\E[\pi_m]{ \sum_i \lambda_{m, i} \kb{e_{\pi_m(i)}}} \\
		&= \E[\pi_m]{ \sum_i \lambda_{m, \pi_m^{-1}(i)} \kb{e_{i}}}\\
		&=  \sum_i \E[\pi_m]{ \lambda_{m, \pi_m^{-1}(i)}} \kb{e_{i}}\\
		&= \sum_{i\in \calS_m} \frac{1}{\card{\calS_m}} \kb{e_{i}}, \label{eq:expected-dens}
\end{align}}
where $(b)$ follows since for $m\neq m'$, $\calM_{V_{\pi_{m'}}}(\rho) = \rho $ for any $\rho$ with support in the span of $\set{\ket{e_i}: i\in \calS_m}$.
Upon defining unitary $v_k$ as $v_k \ket{a_i} = \ket{e_i}, ~\forall i\in [d]$,
we can re-write Eq.~\eqref{eq:expected-dens} as
\needspace{\begin{align}
		\underset{\pi_1, \cdots, \pi_{M}}{\mathbb E}\!\left((\calM_{V_{\pi_1}} \!\!\!\circ \!\cdots \!\circ \!\calM_{V_{\pi_M}} \!\!\!\circ \!\Enc_k)(m)\right) \!=\!v_k\pr{ \sum_{i\in \calS_m} \frac{1}{\card{\calS_m}} \kb{a_{i}}} v_k^\dagger \!= \!\calM_{v_k}\pr{ \sum_{i\in \calS_m} \frac{1}{\card{\calS_m}} \kb{a_{i}}}.  \label{eq:expected-dens2}
\end{align}}
{\begin{align}
    \E[\pi_1, \cdots, \pi_{M}]{(\calM_{V_{\pi_1}} \circ \cdots \circ \calM_{V_{\pi_M}} \circ \Enc_k)(m)}  = v_k\pr{ \sum_{i\in \calS_m} \frac{1}{\card{\calS_m}} \kb{a_{i}}} v_k^\dagger = \calM_{v_k}\pr{ \sum_{i\in \calS_m} \frac{1}{\card{\calS_m}} \kb{a_{i}}}.  \label{eq:expected-dens2}
\end{align}}
Combining Eq.~\eqref{eq:delta-perm-bas} and Eq.~\eqref{eq:expected-dens2}, we obtain
\needspace{\begin{align}
    &\E[U]{ \Delta\pr{\calN_{A\to BC} \circ \calM_U \circ \Enc_k }} 
    \geq {\E[U]{  \Delta\pr{\calN_{A\to BC} \circ \calM_{U} \circ \calM_{v_k} \circ \widehat{\Enc}_{t(k)}}}}\\
    =& {\E[U]{  \Delta\pr{\calN_{A\to BC} \circ \calM_{Uv_k}  \circ \widehat{\Enc}_{t(k)}}}}= {\E[U]{  \Delta\pr{\calN_{A\to BC} \circ \calM_{U}  \circ \widehat{\Enc}_{t(k)}}}},\label{eq:haar-delta2}\quad\text{where}\\
    &\widehat{\Enc}_t(m)\eqdef \frac{1}{t_m} \sum_{i=\sum_{j=1}^{m-1}t_j}^{\sum_{j=1}^{m}t_j} \kb{a_i}.
\end{align}}
{\begin{align}
		\E[U]{ \Delta\pr{\calN_{A\to BC} \circ \calM_U \circ \Enc_k }} 
		&\geq {\E[U]{  \Delta\pr{\calN_{A\to BC} \circ \calM_{U} \circ \calM_{v_k} \circ \widehat{\Enc}_{t(k)}}}}\\
		&= {\E[U]{  \Delta\pr{\calN_{A\to BC} \circ \calM_{Uv_k}  \circ \widehat{\Enc}_{t(k)}}}}\\
		&= {\E[U]{  \Delta\pr{\calN_{A\to BC} \circ \calM_{U}  \circ \widehat{\Enc}_{t(k)}}}},\label{eq:haar-delta2}
\end{align}
where 
\begin{align}
    \widehat{\Enc}_t(m)\eqdef \frac{1}{t_m} \sum_{i=\sum_{j=1}^{m-1}t_j}^{\sum_{j=1}^{m}t_j} \kb{a_i}.
\end{align}}
Putting Eq.~\eqref{eq:win-delta}, Eq.~\eqref{eq:haar-delta}, and Eq.~\eqref{eq:haar-delta2} together, we have
\needspace{\begin{align}
    &\pwinunifs{\calE}  
    \geq \sup_{B, C, \calN_{A\to BC}}\E[k\leftarrow \KeyGen]{\E[U]{  \Delta\pr{\calN_{A\to BC} \circ \calM_{U}  \circ \widehat{\Enc}_{t(k)}}}}\\
    =& \sup_{B, C, \calN_{A\to BC}}\E[T]{\E[U]{  \Delta\pr{\calN_{A\to BC} \circ \calM_{U}  \circ \widehat{\Enc}_{T}}}}\stackrel{(c)}{=} \pwinunifs{\widetilde{\calE_{a, T}}}  
\end{align}}
{\begin{align}
		\pwinunifs{\calE}  
		&\geq \sup_{B, C, \calN_{A\to BC}}\E[k\leftarrow \KeyGen]{\E[U]{  \Delta\pr{\calN_{A\to BC} \circ \calM_{U}  \circ \widehat{\Enc}_{t(k)}}}}\\
		&= \sup_{B, C, \calN_{A\to BC}}\E[T]{\E[U]{  \Delta\pr{\calN_{A\to BC} \circ \calM_{U}  \circ \widehat{\Enc}_{T}}}}\\
		&\stackrel{(c)}{=} \pwinunifs{\haarscheme{T}{d}}  
\end{align}}
where $(c)$ follows because $\calM_{u}  \circ \widehat{\Enc}_{t}(m) = \widetilde{\Enc}_{(t,u)}(m)$. We next show that we can choose $T$ to be permutation invariant. We first define for  $t= (t_1, \cdots, t_M)$ and a permutation $\sigma$ of $[M]$, $\sigma t \eqdef (t_{\sigma(1)}, \cdots, t_{\sigma(M)})$. We then have
\begin{align}
    \Delta\pr{\calN_{A\to BC} \circ \calM_{U}  \circ \widehat{\Enc}_{\sigma t}   } = \Delta\pr{\calN_{A\to BC} \circ \calM_{U}  \circ \widehat{\Enc}_{t} \circ \sigma  } = \Delta\pr{\calN_{A\to BC} \circ \calM_{U}  \circ \widehat{\Enc}_{t}  }
\end{align}
because by definition of $\Delta(\cdot)$, it holds that $\Delta(g \circ \sigma) = \Delta(g)$ for all $g:[M]\to \calD(BC)$ and permutations $\sigma$. We next define random variable $\widetilde{T}$ such that
\begin{align}
    \P{\widetilde{T} = t} = \frac{1}{M!}\sum_{\sigma} \P{T = \sigma t},
\end{align}
which is a valid permutation invariant probability distribution.
Furthermore, 
\begin{align}
    \pwinunifs{\haarscheme{{T}}{d,M}}
    &= \sup_{B, C, \calN_{A\to BC}}{\E[U]{ \sum_{t}\P{T=t} \Delta\pr{\calN_{A\to BC} \circ \calM_{U}  \circ \widehat{\Enc}_{t}}}}\\
    &= \sup_{B, C, \calN_{A\to BC}}{\E[U]{ \sum_{t}\P{T=t}\frac{1}{M!}\sum_{\sigma} \Delta\pr{\calN_{A\to BC} \circ \calM_{U}  \circ \widehat{\Enc}_{\sigma t}}}}\\
    &= \sup_{B, C, \calN_{A\to BC}}{\E[U]{ \sum_{t}\P{T=t}\frac{1}{M!}\sum_{\sigma} \Delta\pr{\calN_{A\to BC} \circ \calM_{U}  \circ \widehat{\Enc}_{\sigma t}}}}\\
    &= \sup_{B, C, \calN_{A\to BC}}{\E[U]{ \sum_{t}\P{\widetilde{T}=t}\Delta\pr{\calN_{A\to BC} \circ \calM_{U}  \circ \widehat{\Enc}_{t}}}}\\
    &=     \pwinunifs{\haarscheme{\widetilde{T}}{d,M}},
\end{align}
as desired.

\section{The Erlang distribution}
\label{sec:erlang-dist}
In this appendix, we gather some properties of the \href{https://en.wikipedia.org/wiki/Erlang_distribution}{Erlang distribution} used in our work. 
\begin{definition}
$\textnormal{Erlang}(k, \lambda)$ is a probability distribution over $[0, \infty)$ characterized by two parameters $k\in \calN$ and $\lambda > 0$ and defined by its probability density function as
\begin{align}
    f(x) &\eqdef \frac{\lambda^kx^{k-1} e^{-\lambda x}}{(k-1)!} \, .
\end{align}
\end{definition}
When $X\sim \textnormal{Erlang}(k, \lambda)$, we have
\begin{align}
    \E{X} &= k/\lambda\\
    \P{X \leq x} &= 1- \sum_{i=0}^{k-1} \frac{e^{-\lambda x} (\lambda x)^i }{i!}.
\end{align}
Moreover, if $X_0, \cdots, X_{2k-1}$ are independent standard normal random variables, then $\sum_{i\in[2k]}|X_i|^2$ has $\textnormal{Erlang}(k, \frac{1}{2})$ distribution. 

\begin{lemma}
\label{lm:erlang-max}
Let $X_0, \cdots, X_{n-1}$ be independent such that $X_i\sim \textnormal{Erlang}(k_i, \lambda)$. We have
\begin{align}
    \E{\frac{\max_{i\in [n]} X_i}{\sum_{i\in[n]} X_i}} \geq c\frac{\log n}{\sum_{i\in[n]} k_i}.
\end{align}
for some absolute constant $c > \frac{1 - e^{-1} -0.5 }{2\log(e)} \approx 0.0457$.
\end{lemma}

\begin{proof}
We have for all $a, b>0$
\begin{align}
    \E{\frac{\max_{i\in [n]} X_i}{\sum_{i\in[n]} X_i}} 
    &\geq \frac{a}{b} \P{\max_{i\in [n]} X_i \geq a \text{ and } \sum_{i\in[n]} X_i \leq b}\\
    &\stackrel{(i)}{\geq} \frac{a}{b}\pr{1 -  \P{\max_{i\in [n]} X_i < a} - \P{\sum_{i\in[n]} X_i > b}}\\
    &\stackrel{(ii)}{\geq}  \frac{a}{b}\pr{1 -  \P{\max_{i\in [n]} X_i < a} - \frac{\sum_{i\in[n]}\E{X_i}}{b}}\\
    &=   \frac{a}{b}\pr{1 -  \P{\max_{i\in [n]} X_i < a} - \frac{\sum_{i\in[n]} k_i}{b\lambda}}
\end{align}
where $(i)$ follows from the union bound and $(ii)$ follows from Markov's inequality. Furthermore, by the independence of $X_1, \ldots, X_n$, we have
\begin{align}
    \P{\max_{i\in [n]} X_i < a} 
    &= \prod_{i\in[n]}  \P{X_i < a}\\
    &= \prod_{i\in[n]} \pr{1-e^{-\lambda a} \sum_{j=0}^{k_i-1} \frac{ (\lambda a)^j }{j!}}\\
    &\leq \pr{1-e^{-\lambda a}}^n \\ 
    &= e^{\ln(1-e^{-\lambda a})n}\\
    &\leq e^{-e^{-\lambda a}n}.
\end{align}
Setting $a=\ln n/\lambda$ and $b= \frac{2\sum_i k_i}{\lambda}$, we obtain that
\begin{align}
    \E{\frac{\max_{i\in [n]} X_i}{\sum_{i\in[n]} X_i}} 
    &\geq \frac{1 - e^{-1} -0.5 }{2}\frac{\ln n}{\sum k_i}\\
    &= \frac{1 - e^{-1} -0.5 }{2\log(e)}\frac{\log n}{\sum k_i}.
\end{align}
\qed\end{proof}

\section{Connection between uncloneable encryption and monogamy of entanglement game}

A MEG $\calG$ has three players Alice, Bob, and Charlie and is described by a set $\calM$, a random variable $K$ with distribution $P_K$, a Hilbert space $A$, and a POVM $\set{F_m^k}_{m\in \calM}$ for each realization $k$ of $K$. Bob and Charlie prepare a tripartite quantum state $\rho_{ABC}$ and pass sub-system $A$ to Alice. Bob and Charlie keep sub-systems $B$ and $C$, respectively, and no communication is allowed between them after state preparation. Alice performs POVM $\set{F_m^k}_{m\in \calM}$ on sub-system $A$ when $K=k$ and provides $K$ to both Bob and Charlie. Bob and Charlie then perform POVMs $\set{P_m^k}_{m\in \calM}$ and $\set{Q_m^k}_{m\in \calM}$, respectively, on their corresponding sub-system. Bob and Charlie win if all three parties obtain the same outcome from their measurements. The probability of winning is 
\begin{align}
    \E[k\sim P_K]{\sum_m\text{tr}(F_m^k\otimes P_m^k\otimes Q_m^k\rho_{ABC})}.
\end{align}
We denote by $\textnormal{p}_{\textnormal{win}}^*(\mathcal{G})$ the maximum probability of winning that Bob and Charlie can achieve for a fixed game $\calG$. In the next two lemmas, we establish lower-bounds on the probability of winning a MEG.

\begin{proposition}
Let $\calE = (\KeyGen, \Enc, \Dec)$ be a correct QECM with $M$ messages such that $\frac{1}{M}\sum_{m\in[M]}\Enc_k(m) =  \overline{\rho}$ independent of $k$. We define a MEG $\calG$ with the same key $k$ as $\calE$ and $F_m^k = \frac{1}{M} \overline{\rho}^{-\frac{1}{2}} \Enc_k(m)\overline{\rho}^{-\frac{1}{2}}  $ where the transpose is with respect to eigenvectors of $\overline{\rho}$. We then have
\begin{align}
    \pwinunifs{\calE} \leq \textnormal{p}_{\textnormal{win}}^*(\mathcal{G}).
\end{align}
\end{proposition}

\begin{proof}
We fix an arbitrary cloning attack $\calA = (\calN_{A\to BC}, \set{P_m^k}, \set{Q_m^k})$ on $\calE$ and show that $\pwinunif{\calE}{\calA} \leq \textnormal{p}_{\textnormal{win}}^*(\mathcal{G})$. Let $\overline{\rho} = \sum_{i\in[d]} \lambda_i \kb{e_i}$ where $e = (\ket{e_0}, \cdots, \ket{e_{d-1}})$ is an orthonormal basis  for $A$. We define $\ket{\Phi} \eqdef \sum_{i\in[d]}\sqrt{\lambda_i} \ket{e_i}\otimes \ket{e_i}$ and the Choi isomorphism $J$ with respect to $\overline{\rho}$ that maps a quantum channel $\calN_{A\to BC}$ to a density operator 
\begin{align}
    (\id_A \otimes \calN_{A\to BC})(\kb{\Phi}) = \sum_{i\in[d]} \sum_{j\in[d]}\sqrt{\lambda_i \lambda_j} \ket{e_i}\bra{e_j} \otimes \calN_{A\to BC}(\ket{e_i}\bra{e_j}).
\end{align}
We have for all operators $X$ acting on $A$
\begin{align}
    \textnormal{tr}_{A}\pr{\pr{(\overline{\rho}^{-\frac{1}{2}} X^T\overline{\rho}^{-\frac{1}{2}}) \otimes \one_{BC}} J(\calN)} 
    &= \sum_{i\in[d], j\in[d]}   \sqrt{\lambda_i \lambda_j}  \textnormal{tr}_{A}\pr{\pr{(\overline{\rho}^{-\frac{1}{2}} X^T\overline{\rho}^{-\frac{1}{2}}) \otimes \one_{BC}} \ket{e_i}\bra{e_j} \otimes \calN_{A\to BC}(\ket{e_i}\bra{e_j})} \\
    &=\sum_{i\in[d], j\in[d]}   \sqrt{\lambda_i \lambda_j} \tr{\overline{\rho}^{-\frac{1}{2}} X^T\overline{\rho}^{-\frac{1}{2}}\ket{e_i}\bra{e_j}} \calN_{A\to BC}(\ket{e_i}\bra{e_j})\\
    &= \sum_{i\in[d], j\in[d]} \bra{e_i} X \ket{e_j}\calN_{A\to BC}(\ket{e_i}\bra{e_j}) = \calN\pr{X}.
\end{align}
 We then have 
\needspace{\begin{align}
     &\E[k\leftarrow \KeyGen]{\sum_{m\in \calM}\text{tr}(F_m^k\otimes P_m^k\otimes Q_m^kJ(\calN_{A\to BC}))}\\
     &= \frac{1}{\kappa}\sum_{k\in [\kappa]}\sum_{m\in \calM}\tr{F_m^k\otimes P_m^k\otimes Q_m^k\pr{    \frac{1}{d}\sum_{i=1}^d \sum_{j=1}^d \ket{e_i}\bra{e_j} \otimes \calN_{A\to BC}(\ket{e_i}\bra{e_j})}}\\
     &= \frac{1}{d\kappa}\sum_{k\in [\kappa]}\sum_{m\in \calM}\sum_{i=1}^d\sum_{j=1}^d  \bra{e_j}F_m^k \ket{e_i}\tr{ P_m^k\otimes Q_m^k \calN_{A\to BC}(\ket{e_i}\bra{e_j})}\\
     &= \frac{1}{d\kappa}\sum_{k\in [\kappa]}\sum_{m\in \calM}\tr{ P_m^k\otimes Q_m^k \calN_{A\to BC}(\Pi_{m}^k)}=\frac{1}{\card{\calM}\kappa} \sum_{m\in \calM}  \sum_k\text{tr}(P_m^k\otimes Q_m^k \calN_{A\to BC}(\Enc_k(m))).
\end{align}}
{\begin{align}
        &\pwinunif{\calE}{\calA} = \E[k\leftarrow \KeyGen]{\frac{1}{M} \sum_{m\in [M]}  \tr{P_m^k\otimes Q_m^k \calN_{A\to BC}(\Enc_k(m)})}\\
        &~~~~~~~~~~~~~~~~~~~~~=\E[k\leftarrow \KeyGen]{\frac{1}{M} \sum_{m\in [M]}  \tr{P_m^k\otimes Q_m^k               \textnormal{tr}_{A}\pr{\pr{(\overline{\rho}^{-\frac{1}{2}} \Enc_k(m)^T\overline{\rho}^{-\frac{1}{2}}) \otimes \one_{BC}} J(\calN)} }}\\
        &~~~~~~~~~~~~~~~~~~~~~=\E[k\leftarrow \KeyGen]{\frac{1}{M} \sum_{m\in [M]}  \tr{\one_A \otimes P_m^k\otimes Q_m^k  {\pr{(\overline{\rho}^{-\frac{1}{2}} \Enc_k(m)^T\overline{\rho}^{-\frac{1}{2}}) \otimes \one_{BC}} J(\calN)} }}\\
		&~~~~~~~~~~~~~~~~~~~~~=\E[k\leftarrow \KeyGen]{\sum_{m\in [M]}\text{tr}(F_m^k\otimes P_m^k\otimes Q_m^kJ(\calN))} \leq \textnormal{p}_{\textnormal{win}}^*(\mathcal{G}),
\end{align}}
as claimed.

\end{proof}

\end{document}